\documentclass[11pt,a4paper]{article}

\usepackage[utf8]{inputenc}
\usepackage[T1]{fontenc}
\usepackage{lmodern}
\usepackage{amsmath,amssymb,amsthm}
\usepackage{mathtools}
\usepackage{forest} 
\usepackage{enumitem}
\usepackage{hyperref}
\usepackage{microtype}
\usepackage{geometry}

\usepackage{algorithm}
\usepackage{algpseudocode}
\usepackage{float}

\usepackage{authblk}
\usepackage{orcidlink} 

\title{Templated Assembly Theory:\\
An Extension of the Canonical Assembly Index with Block-Compressed Templates}
\author{Piotr Masierak\orcidlink{0009-0008-9895-5329}}
\affil{Łukaszyk Patent Attorneys, ul. Głowackiego 8, 40-052 Katowice, Poland\\
\texttt{pmasierak@patent.pl}}
\date{January 26, 2026}

\theoremstyle{definition}
\newtheorem{definition}{Definition}[section]
\theoremstyle{plain}

\newtheorem{lemma}[definition]{Lemma}
\newtheorem{proposition}[definition]{Proposition}
\theoremstyle{remark}
\newtheorem{remark}[definition]{Remark}

\newcommand{\ASI}{\mathrm{ASI}}
\newcommand{\TAI}{\mathrm{TAI}}

\begin{document}

\maketitle

\begin{abstract}
Assembly Theory, as developed by Cronin and co-workers, assigns to an object an assembly index: the minimal number of binary join operations required to build at least one copy of the object from a specified set of basic building blocks, allowing reuse of intermediate components.
For strings over a finite alphabet, the canonical assembly index can be defined in the free semigroup $(\Sigma^{+},\cdot)$ with universal binary concatenation and a ``no-trash'' condition, and its exact computation has been shown to be NP-complete.
In this paper we propose an extension of the canonical, string-based formulation which augments pure concatenation with templated assembly steps.
Intermediate objects may contain a distinguished wildcard symbol $\ast$ that represents a compressible block.
Templates are restricted to block-compressed substrings of the target string and can be instantiated by inserting previously assembled motifs into one or many wildcard positions, possibly in parallel.
This yields a new complexity measure, the templated assembly index, which strictly generalises the canonical index while preserving its operational character.
We formalise the model, clarify its relation to the canonical assembly index and to classical problems such as the smallest grammar problem, and discuss the computational complexity of determining the templated assembly index.
Finally, we sketch potential applications in sequence analysis, modularity detection, and biosignature design.
\end{abstract}

\section{Introduction}

Assembly Theory (AT) aims to quantify the historical contingency and structural complexity of objects by counting the minimal number of composition steps required to build them from simple components, under a specified set of admissible joins and with reuse of previously assembled substructures 
~\cite{cronin2017,cronin2022,cronin2023}.
In its string-based form, one considers the free semigroup $(\Sigma^{+},\cdot)$ over a finite alphabet $\Sigma$, with binary concatenation as the only join operation.
The canonical assembly index of a string $w\in\Sigma^{+}$ is then the minimal length of an assembly plan that produces at least one copy of $w$, using only substrings of $w$ as intermediate objects and charging unit cost per join.
This canonical, ``no-trash'' variant captures a universal lower bound for more operational, domain-specific versions of Assembly Theory.\\

In ~\cite{lukaszyk_assembly_2024} it was conjectured and in ~\cite{masierak_computational_2025} it was shown that computing the 
assembly index of a given string---or even deciding whether the index is at most a given bound $k$---is NP-complete, via a reduction from the smallest grammar problem and a correspondence between assembly plans and straight-line programs.
From a structural perspective, the canonical index is sensitive to hierarchical reuse of substrings, especially repeated motifs that can be assembled once and then copied in subsequent joins.\\

However, many natural systems exhibit forms of templated or pattern-based assembly that are not fully captured by pure concatenation.
For instance, a fixed molecular scaffold may contain multiple sites where the same functional group is attached, or a regulatory DNA sequence may contain recurrent binding motifs embedded in varying local contexts.
From the perspective of string assembly, these phenomena resemble the use of templates: patterns with holes (wildcard positions) that are subsequently filled by previously assembled motifs, possibly at many positions in parallel.\\

The central aim of this paper is to formalise this intuition in a way that remains faithful to the operational ethos of Assembly Theory.
We introduce a new variant of the canonical string-based model in which intermediate objects may contain a special symbol $\ast$ denoting a placeholder for an arbitrary non-empty string.
Crucially, we restrict templates to those that arise as block-compressed substrings of the target string:
roughly, a template is obtained by taking a substring $u$ of the target $w$ and replacing some contiguous blocks of $u$ by $\ast$, thereby compressing entire blocks into a single wildcard.
Assembly operations now include not only concatenation, but also template instantiation steps, where a previously available motif $x$ is substituted for one or all occurrences of $\ast$ in a template $T$, in a single step.\\

This yields a new complexity measure, the templated assembly index $\TAI(w)$.
By design, the canonical index $\ASI(w)$ appears as the special case where wildcard templates are never used, and we show that $\TAI(w)\le \ASI(w)$ for all $w$.
At the same time, $\TAI(w)$ can be strictly smaller than $\ASI(w)$ for strings that exhibit non-local, templated reuse of motifs, making $\TAI$ sensitive to a different aspect of structure---which we term templated modularity.\\

\section{Background: canonical string-based Assembly Theory}

We briefly recall the canonical string-based formulation of Assembly Theory and fix notation.
Let $\Sigma$ be a finite alphabet and let $\Sigma^{+}$ denote the set of all non-empty strings over $\Sigma$.
We consider the free semigroup $(\Sigma^{+},\cdot)$ with concatenation as the binary operation.

\begin{definition}
Let $w\in\Sigma^{+}$.
The canonical assembly space for $w$ is the pair
\[
\mathcal{A}(w) = (O,R),
\]
where:
\begin{itemize}
    \item $O\subseteq \Sigma^{+}$ is the set of objects, taken to be the set $\mathrm{Sub}(w)$ of all non-empty substrings of $w$;
    \item $R$ is the set of binary join rules of the form
    \[
    x \circ y \to xy
    \]
    for all $x,y\in O$ such that $xy\in O$.
\end{itemize}
\end{definition}

Informally, $O$ contains exactly the substrings of the target $w$, and we may join two substrings $x,y$ in a single assembly step whenever their concatenation $xy$ is again a substring of $w$.
Intermediate objects are thus constrained by a no-trash condition: every object appearing in any assembly plan is a substring of the target $w$.

\begin{definition}
Let $w\in\Sigma^{+}$ and $\mathcal{A}(w)=(O,R)$ as above.
An assembly plan for $w$ is a sequence
\[
x_1, x_2, \dots, x_m
\]
of objects in $O$ such that:
\begin{enumerate}[label=(\alph*)]
    \item each $x_i$ is either a single letter in $\Sigma$ or is obtained from two earlier objects by a rule in $R$, i.e.\ there exist $j,k<i$ with $x_i = x_j x_k$ and $x_j x_k \in O$;
    \item at least one occurrence of $w$ appears among the $x_i$.
\end{enumerate}
The length of the assembly plan is the number of composite steps, i.e.\ the number of indices $i$ for which $x_i$ is obtained by some join rule in $R$.
The canonical assembly index of $w$, denoted $\ASI(w)$, is the minimal length of an assembly plan for $w$.
\end{definition}

This formulation captures the intuition that we start from monomeric building blocks (single letters) and in each step join two previously available objects, paying unit cost for each join, until we have built the target $w$ at least once.
Crucially, the reuse principle is built in: once an object $x$ has appeared in the assembly plan, it is available for use in arbitrary many subsequent joins, effectively at zero marginal cost.

\section{Templated assembly spaces}

We now introduce a templated extension of the canonical assembly space model.
Throughout this section we fix a finite alphabet $\Sigma$ and a distinguished wildcard symbol $\ast\notin\Sigma$.\\

Intuitively, a template is obtained from a substring of the target word by compressing whole contiguous blocks into a single wildcard symbol $\ast$.
The wildcard positions will later be instantiated by inserting previously assembled motifs.

\begin{definition}
Let $w\in\Sigma^{+}$.
A string $T\in(\Sigma\cup\{\ast\})^{+}$ is called a block-compressed template for $w$ if there exist indices $0\le i\le j<|w|$ and a partition
\[
0 = b_0 < b_1 < \dots < b_k = L
\]
of the positions of the substring $u = w[i..j]$ of length $L = j-i+1$ into contiguous blocks
\[
S_r = u[b_r..b_{r+1}-1] \in \Sigma^{+},\quad r=0,\dots,k-1,
\]
together with a bit vector $F\in\{0,1\}^k$ such that:
\begin{itemize}
    \item $F$ is neither all zeros nor all ones;
    \item $T$ is obtained by replacing each block $S_r$ by
    \[
    T_r =
    \begin{cases}
    S_r & \text{if } F[r] = 0,\\
    \ast & \text{if } F[r] = 1,
    \end{cases}
    \]
    and setting $T = T_0 T_1 \cdots T_{k-1}$.
\end{itemize}
We write $\mathcal{T}(w)$ for the set of all block-compressed templates for $w$.
\end{definition}

Note that templates are not assumed to be given independently of the input: for a fixed target string $w$ the only admissible templates are those in $\mathcal{T}(w)$ (i.e.\ mined from substrings of $w$ by block-compression), and any such template can be used in an assembly plan only after it has been explicitly constructed from monomers via rules in $R_{\mathrm{concat}}$ (possibly reusing previously assembled objects).

Thus, $T$ is obtained from some contiguous substring $u$ of $w$ by compressing one or more entire blocks into a single wildcard; at least one block remains literal, and at least one block is compressed.
The wildcard symbol $\ast$ can be thought of as a ``hole'' that marks a region where an arbitrary non-empty string may later be inserted.

Note that block boundaries are arbitrary, so compressed blocks may have length~1. In particular, templates can encode ``single-site'' variability: a wildcard position may correspond to a single symbol in the underlying substring, and instantiating it with different monomers in $\Sigma$ captures patterns that
differ only at specific positions.

\begin{remark}
The restriction to block-compressions of substrings enforces a templated version of the no-trash condition: every template $T$ is anchored in $w$ by its literal segments, which appear in the same order as in $w$.
In particular, any instantiation of $T$ by compressible blocks that themselves occur as substrings of $w$ yields again a substring of $w$.\\
\end{remark}

We now define the objects and operations in a templated assembly space associated with a target string $w$.

\begin{definition}
Let $w\in\Sigma^{+}$.
The templated assembly space for $w$ is the triple
\[
\mathcal{A}_{\ast}(w) = (O_{\ast},R_{\mathrm{concat}},R_{\mathrm{temp}}),
\]
where
\begin{itemize}
    \item $O_{\ast}$ is the set of objects, defined as
    \[
    O_{\ast} = \mathrm{Sub}(w)\; \cup\; \mathcal{T}(w).
    \]
    Thus every object is either a substring of $w$ (containing no wildcards) or a block-compressed template for $w$ (containing at least one wildcard and at least one literal segment).
    \item $R_{\mathrm{concat}}$ is the set of concatenation rules
    \[
    x \circ y \to xy
    \]
    for all $x,y\in O_{\ast}$ such that $xy\in O_{\ast}$.
    
    \item $R_{\mathrm{temp}}$ is the set of template instantiation rules.
    For a template $T\in\mathcal{T}(w)$ let $P(T)$ denote the set of positions of wildcards in $T$, i.e.
    \[
      P(T) \;=\; \{\, p : T[p] = \ast \,\}.
    \]
    For every non-empty subset $S \subseteq P(T)$ and every object $u\in O_{\ast}$
    such that $T[S\mapsto u]\in O_{\ast}$, we have a rule
    \[
      T \bullet_{S} u \;\to\; T[S \mapsto u],
    \]
    where $T[S\mapsto u]$ denotes the string obtained from $T$ by replacing each wildcard at a position in $S$ by $u$ and leaving all other positions unchanged.
    In particular, if $S=\{p\}$ is a singleton, we obtain a single-star substitution at position $p$, and if $S = P(T)$ we obtain a fully parallel substitution that replaces all wildcards in $T$ by $u$ simultaneously.
    As before, we explicitly require $T[S\mapsto u]\in O_{\ast}$; this is automatically satisfied whenever $u$ is a substring of $w$, but in general we impose it as part of the no-trash condition.
\end{itemize}
\end{definition}

\begin{remark}
Each rule application in $R_{\mathrm{concat}}\cup R_{\mathrm{temp}}$ has unit cost.

Templates are not external primitives. The only templates that may be used when assembling $w$ are those in $\mathcal{T}(w)$, i.e.\ templates mined from substrings of the target string (Definition above).
Moreover, a template $T\in\mathcal{T}(w)$ can appear in an assembly plan only after it has first been built using the usual binary concatenation rules in $R_{\mathrm{concat}}$ from monomers (letters in $\Sigma$ and the wildcard~$\ast$).
In particular, the construction of a template may reuse any objects that have already been constructed earlier in the same plan (including substrings and previously assembled templates).

Thus the total cost of ``constructing'' a template is exactly the number of concatenation steps used in its assembly, while each subsequent template instantiation step from $R_{\mathrm{temp}}$ is counted as a separate unit-cost operation.\\
\end{remark}

\begin{remark}
The instantiation rules $T\bullet_{S}u \to T[S\mapsto u]$ allow both single-star substitutions (when $|S|=1$) and fully parallel substitutions (when $S=P(T)$) as special cases. In this sense, the present templated model subsumes both the single-star and all-star variants.\\
\end{remark}

Thus, assembly operations now include both ordinary concatenation of objects and templated insertion of previously assembled motifs into block-compressed templates.
We stress that templates themselves are not introduced by a special rule: they are statically determined by $w$ via block-compression of substrings and are part of $O_{\ast}$ from the outset.
All dynamic operations are binary and of unit cost, consistent with the canonical model.\\

We next define assembly plans in the templated setting and the associated complexity measure.\\

\begin{definition}
Let $w\in\Sigma^{+}$ and $\mathcal{A}_{\ast}(w) = (O_{\ast},R_{\mathrm{concat}},R_{\mathrm{temp}})$.
A templated assembly plan for $w$ is a finite sequence
\[
x_1,x_2,\dots,x_m
\]
of objects in $O_{\ast}$ such that:
\begin{enumerate}[label=(\alph*)]
    \item each $x_i$ is either a monomer (a single letter in $\Sigma$ or the wildcard $\ast$) or is obtained by applying a rule in $R_{\mathrm{concat}}\cup R_{\mathrm{temp}}$ to earlier objects;
    \item at least one occurrence of $w$ appears among the $x_i$.
\end{enumerate}
The length of the assembly plan is the number of composite steps, i.e.\ the number of indices $i$ for which $x_i$ is obtained via a rule in $R_{\mathrm{concat}}\cup R_{\mathrm{temp}}$.\\
\end{definition}

We are now ready to define the templated assembly index.

\begin{definition}
Let $w\in\Sigma^{+}$.
The templated assembly index of $w$, denoted $\TAI(w)$, is the minimal length (number of composite steps) among all templated assembly plans for $w$ in $\mathcal{A}_{\ast}(w)$.
\end{definition}

By construction, we may always ignore templates and restrict to substrings and concatenation, yielding a canonical assembly plan.

\begin{proposition}
For all $w\in\Sigma^{+}$ we have
\[
\TAI(w) \;\le\; \ASI(w).
\]
Moreover, if $w$ admits no templated assembly plan that uses a rule in $R_{\mathrm{temp}}$, then $\TAI(w) = \ASI(w)$.
\end{proposition}

\begin{proof}[Sketch]
Any canonical assembly plan in $\mathcal{A}(w)$ is also a valid templated assembly plan in $\mathcal{A}_{\ast}(w)$: the set of substrings $\mathrm{Sub}(w)$ is contained in $O_{\ast}$, and all canonical concatenation rules appear in $R_{\mathrm{concat}}$.
Hence $\TAI(w)$, as a minimum over a larger set of assembly plans, cannot exceed $\ASI(w)$.

If no templated rule in $R_{\mathrm{temp}}$ can be applied along any optimal templated assembly plan without violating the object constraint ($O_{\ast}$), then any optimal templated assembly plan uses only concatenation rules; restricting to substrings recovers a canonical assembly plan of the same length, and thus $\TAI(w)=\ASI(w)$.
\end{proof}

In practice, $\TAI(w)$ can be strictly smaller than $\ASI(w)$ for strings that contain repeated motifs embedded in non-local contexts describable by block-compressed templates.
We illustrate this effect in the next subsection.\\

\begin{remark}
More generally, let $O$ be a fixed object universe and let $R_1 \subseteq R_2$ be two sets of assembly rules on $O$.
For $i\in\{1,2\}$, let $\ASI_i(w)$ denote the minimal length of an assembly plan for $w$ when only rules from $R_i$ are allowed.
Then for every $w$ we have $\ASI_2(w) \le \ASI_1(w)$, since any $R_1$-plan is also a valid $R_2$-plan.
In particular, any more restrictive templated model obtained by removing some instantiation rules from $R_{\mathrm{temp}}$ would give rise to an index $\TAI'(w)$ satisfying
\[
  \TAI(w) \;\le\; \TAI'(w) \;\le\; \ASI(w).
\]
Thus the present $\TAI(w)$ can be seen as the smallest member of a family of canonical indices obtained by varying the strength of the templated rule set.
\end{remark}

\subsection{Illustrative example. Selective template mining and when wildcards help}

A key point of the present model is that templates are mined from the target string $w$:
the admissible templates are exactly the block-compressed templates in $\mathcal{T}(w)$.
The index $\TAI(w)$ is defined by minimising over all templated assembly plans in $\mathcal{A}_{\ast}(w)$.
Consequently, templates are used only if they are beneficial:
if no plan that employs a rule in $R_{\mathrm{temp}}$ improves upon the best purely concatenative plan, then an optimal plan simply refrains from using templates and we obtain $\TAI(w)=\ASI(w)$ (Proposition above).\\

To understand when the wildcard mechanism can reduce the number of steps, it is useful to separate two issues:
(i) the cost of assembling a template skeleton (counted via $R_{\mathrm{concat}}$), and
(ii) the cost of instantiating wildcard sites (counted via $R_{\mathrm{temp}}$).

\begin{itemize}
  \item \textbf{Single-use templates.}
  If a template $T$ is assembled and instantiated only once, then introducing it typically does not help:
  since $T$ is not a primitive object, one must pay (via $R_{\mathrm{concat}}$) for constructing its skeleton, and then pay an additional unit-cost instantiation step in $R_{\mathrm{temp}}$.
  In particular, when $\lvert P(T)\rvert = 1$ (i.e.\ $T$ has a single wildcard position), an instantiation step   provides no parallelism across multiple sites; any advantage over purely concatenative assembly is therefore   not automatic and depends on how the surrounding context is reused elsewhere in the plan.
  \item \textbf{Amortisation by reuse.}
  Templates can become beneficial once the same assembled skeleton is instantiated multiple times within an assembly plan (possibly with different fillers). In that case, the one-time cost of constructing $T$ can be amortised across several instantiations. This effect can already occur for $\lvert P(T)\rvert = 1$, but it   requires repeated use of the same contextual ``frame'' represented by $T$.
  \item \textbf{Parallel filling of several sites.}
  Even for a fixed target $w$, a single instantiation rule may fill several wildcard positions at once   (choose $S\subseteq P(T)$). Thus, templates with $\lvert P(T)\rvert\ge 2$ enable a genuinely new mechanism:
  one operation in $R_{\mathrm{temp}}$ can realise multiple occurrences of a motif simultaneously, which is   the main source of potential step savings relative to purely concatenative assembly.
\end{itemize}

As a schematic example, consider a string of the form
\[
w \;=\; P\;S,
\qquad
P = a_0\,u\,a_1\,u\,a_2\cdots a_{m-1}\,u\,a_m,
\qquad
S = a_0\,v\,a_1\,v\,a_2\cdots a_{m-1}\,v\,a_m,
\]
where $a_0,\dots,a_m\in\Sigma^{+}$ are fixed ``anchors'' and $u,v\in\Sigma^{+}$ are two motifs.
The template skeleton
\[
T \;=\; a_0\,\ast\,a_1\,\ast\,a_2\cdots a_{m-1}\,\ast\,a_m
\]
is an element of $\mathcal{T}(w)$ (it is mined from $w$ by compressing each occurrence of the motif block), and after constructing $T$ once, we may instantiate it twice to obtain $P$ and $S$.
When the anchor context is large and the number of wildcard sites is at least two, this may yield a shorter overall plan than assembling $P$ and $S$ separately by concatenation alone.
The worked example in the next subsection provides a concrete instance of this separation mechanism.\\
\subsection{A worked example: separating $\ASI$ and $\TAI$}
\label{subsec:worked-example-asi-tai}

Considering the string
\[
w = \texttt{11221122110001100110011}.
\]
we now present a concrete worked example over the alphabet $\Sigma = \{0,1,2\}$ which illustrates that the templated assembly index $\TAI(w)$ can be strictly smaller than the canonical assembly index $\ASI(w)$ for the same target string.

\paragraph{Canonical assembly plan (showing $ASI(w)=12$).}

In the canonical assembly space $\mathcal{A}(w)$, we restrict attention to substrings of $w$ and binary concatenation.
One possible assembly plan for $w$ proceeds as follows (each line is one unit-cost step):
\begin{align*}
\texttt{00} &= \texttt{0} + \texttt{0},\\
\texttt{11} &= \texttt{1} + \texttt{1},\\
\texttt{22} &= \texttt{2} + \texttt{2},\\
\texttt{1100} &= \texttt{11} + \texttt{00},\\
\texttt{2211} &= \texttt{22} + \texttt{11},\\
\texttt{22110} &= \texttt{2211} + \texttt{0},\\
\texttt{001100} &= \texttt{00} + \texttt{1100},\\
\texttt{110011} &= \texttt{1100} + \texttt{11},\\
\texttt{112211} &= \texttt{11} + \texttt{2211},\\
\texttt{11221122110} &= \texttt{112211} + \texttt{22110},\\
\texttt{001100110011} &= \texttt{001100} + \texttt{110011},\\
\texttt{11221122110001100110011} &= \texttt{11221122110} + \texttt{001100110011}.
\end{align*}
This gives a canonical assembly plan of length $12$. 
We have verified computationally that no canonical assembly plan
of length at most $11$ exists. Hence
\[
ASI(w) = 12.
\]

\paragraph{Templated assembly plan and $\TAI(w)$.}

For readability, we decompose the string $w$ into three contiguous blocks
\[
w = P \, M \, S,
\]
where
\[
P = \texttt{1122112211}, \qquad
M = \texttt{000}, \qquad
S = \texttt{1100110011}.
\]
Both $P$ and $S$ share the same ``anchor'' skeleton
\[
\texttt{11} \;\square\; \texttt{11} \;\square\; \texttt{11},
\]
in which the placeholder $\square$ is realised as \texttt{22} in $P$ and as \texttt{00} in $S$.
This is precisely the situation where block-compressed templates with a wildcard symbol $\ast$ capture a non-local form of modularity.
Therefore in the templated assembly space $\mathcal{A}_{\ast}(w)$ we may additionally use the wildcard symbol $\ast$ and the block-compressed template
\[
T = \texttt{11*11*11},
\]
which is obtained by block-compressing the substring \texttt{1100110011} of $w$ into three anchor blocks \texttt{11} separated by two wildcards.
Starting from monomers, we can build $T$ via the following steps:
\begin{align*}
\texttt{11} &= \texttt{1} \circ \texttt{1},\\
\texttt{11*} &= \texttt{11} \circ \ast,\\
\texttt{11*11*} &= \texttt{11*} \circ \texttt{11*},\\
\texttt{11*11*11} &= \texttt{11*11*} \circ \texttt{11}.
\end{align*}
We also build the motifs \texttt{22} and \texttt{00}:
\begin{align*}
\texttt{22} &= \texttt{2} \circ \texttt{2},\\
\texttt{00} &= \texttt{0} \circ \texttt{0}.
\end{align*}
Using the parallel substitution rule, we obtain the prefix and suffix of $w$ in two single steps:
\begin{align*}
P &= T[\ast \mapsto \texttt{22}] = \texttt{1122112211},\\
S &= T[\ast \mapsto \texttt{00}] = \texttt{1100110011}.
\end{align*}
We construct the middle block and final concatenations as
\begin{align*}
\texttt{000} &= \texttt{00} + \texttt{0},\\
\texttt{1122112211000} &= P + \texttt{000},\\
w &= \texttt{11221122110001100110011} = \texttt{1122112211000} + S.
\end{align*}

Thus the above described templated assembly plan can be written as a step-by-step derivation
(where each line is one unit-cost operation):
\begin{align*}
\texttt{11} &= \texttt{1} + \texttt{1},\\
\texttt{11*} &= \texttt{11} + \ast,\\
\texttt{11*11*} &= \texttt{11*} + \texttt{11*},\\
\texttt{11*11*11} &= \texttt{11*11*} + \texttt{11},\\
\texttt{22} &= \texttt{2} + \texttt{2},\\
\texttt{00} &= \texttt{0} + \texttt{0},\\
P &= \texttt{11*11*11}[\ast \mapsto \texttt{22}] = \texttt{1122112211},\\
S &= \texttt{11*11*11}[\ast \mapsto \texttt{00}] = \texttt{1100110011},\\
\texttt{000} &= \texttt{00} + \texttt{0},\\
\texttt{1122112211000} &= P + \texttt{000},\\
w &= \texttt{11221122110001100110011} = \texttt{1122112211000} + S.
\end{align*}

Counting composite steps, this templated assembly plan has length $11$: four steps to build $T$,
two steps to build \texttt{22} and \texttt{00}, two parallel substitutions to obtain $P$ and $S$,
one step for \texttt{000}, and two final concatenations.
Thus
\[
\TAI(w) \le 11.
\]
Moreover, for this instance the canonical plan given above has length $12$ and this bound is tight, so that
\[
\ASI(w)=12.
\]
Consequently, we obtain a strict separation for the same target string:
\[
\TAI(w) < \ASI(w).
\]

\subsection{Additional worked example I: separating ASI and TAI}
Considering the string
\[
w_1 = 10113121101011212111211,
\]
we present another worked example over the alphabet $\Sigma=\{0,1,2,3\}$ illustrating a strict separation between the canonical assembly index and the templated assembly index.

\paragraph{Canonical assembly plan (showing $\mathrm{ASI}(w_1)=13$).}
In the canonical assembly space $A(w_1)$ we restrict attention to substrings of $w_1$ and binary concatenation. One possible assembly plan for $w_1$ proceeds as follows (each line is one unit-cost step):
\begin{align*}
01 &= 0+1,\\
21 &= 2+1,\\
011 &= 01+1,\\
211 &= 21+1,\\
0113 &= 011+3,\\
1211 &= 1+211,\\
01121 &= 011+21,\\
10113 &= 1+0113,\\
121101 &= 1211+01,\\
2111211 &= 211+1211,\\
10113121101 &= 10113+121101,\\
011212111211 &= 01121+2111211,\\
10113121101011212111211 &= 10113121101+011212111211.
\end{align*}
This gives a canonical assembly plan of length $13$. We have verified computationally that no canonical assembly plan of length at most $12$ exists. Hence
\[
ASI(w_1) = 13.
\]

\paragraph{Templated assembly plan and $\mathrm{TAI}(w_1)$.}
For readability, factor $w_1$ into contiguous blocks
\[
w_1 = 1011\;3\;1211\;0\;1011\;2\;1211\;1211.
\]
The two frequently reused blocks $1011$ and $1211$ share the same anchor skeleton
\[
1\;\square\;11,
\]
where the placeholder $\square$ is realised as $0$ in $1011$ and as $2$ in $1211$.
Accordingly, in the templated assembly space $A_\ast(w_1)$ we use the wildcard symbol $\ast$ and the
block-compressed template
\[
T = 1\ast 11.
\]
A step-by-step templated derivation (each line is one unit-cost operation) is:
\begin{align*}
11 &= 1+1,\\
1\ast &= 1+\ast,\\
T=1\ast 11 &= 1\ast + 11,\\
1011 &= T[\ast\mapsto 0],\\
1211 &= T[\ast\mapsto 2],\\
10113 &= 1011+3,\\
101131211 &= 10113+1211,\\
1011312110 &= 101131211+0,\\
10113121101011 &= 1011312110+1011,\\
101131211010112 &= 10113121101011+2,\\
1011312110101121211 &= 101131211010112+1211,\\
w_1=10113121101011212111211 &= 1011312110101121211+1211.
\end{align*}
Counting steps, we have $3$ steps to build $T$, $2$ instantiation steps to obtain $1011$ and $1211$,
and $7$ concatenations to assemble the final word, for a total of $12$ steps. Thus
\[
\mathrm{TAI}(w_1)\le 12.
\]
With $\mathrm{ASI}(w_1)=13$ this yields a strict separation:
\[
\mathrm{TAI}(w_1)<\mathrm{ASI}(w_1).
\]

\subsection{Additional worked example II: separating ASI and TAI}
Considering the string
\[
w_2 = 101131211010112121112111011,
\]
we obtain an analogous separation.

\paragraph{Canonical assembly plan (showing $\mathrm{ASI}(w_2)=14$).}
One possible canonical assembly plan is (each line is one unit-cost step):
\begin{align*}
01 &= 0+1,\\
10 &= 1+0,\\
11 &= 1+1,\\
12 &= 1+2,\\
0101 &= 01+01,\\
1011 &= 10+11,\\
1211 &= 12+11,\\
10113 &= 1011+3,\\
121211 &= 12+1211,\\
12110101 &= 1211+0101,\\
12111011 &= 1211+1011,\\
1011312110101 &= 10113+12110101,\\
12121112111011 &= 121211+12111011,\\
101131211010112121112111011 &= 1011312110101+12121112111011.
\end{align*}
This gives a canonical assembly plan of length $14$. We have verified computationally that no canonical assembly plan of length at most $13$ exists. Hence
\[
ASI(w_2) = 14.
\]

\paragraph{Templated assembly plan and $\mathrm{TAI}(w_2)$.}
Factor
\[
w_2 = 1011\;3\;1211\;0\;1011\;2\;1211\;1211\;1011.
\]
As before, use the skeleton $1\square 11$ captured by the template $T=1\ast 11$.
A step-by-step templated derivation is:
\begin{align*}
11 &= 1+1,\\
1\ast &= 1+\ast,\\
T=1\ast 11 &= 1\ast + 11,\\
1011 &= T[\ast\mapsto 0],\\
1211 &= T[\ast\mapsto 2],\\
10113 &= 1011+3,\\
101131211 &= 10113+1211,\\
1011312110 &= 101131211+0,\\
10113121101011 &= 1011312110+1011,\\
101131211010112 &= 10113121101011+2,\\
1011312110101121211 &= 101131211010112+1211,\\
10113121101011212111211 &= 1011312110101121211+1211,\\
w_2=101131211010112121112111011 &= 10113121101011212111211+1011.
\end{align*}
Counting steps: $3$ to build $T$, $2$ instantiations, and $8$ concatenations (for $9$ blocks), totalling $13$.
Hence
\[
\mathrm{TAI}(w_2)\le 13,
\]
and with $\mathrm{ASI}(w_2)=14$ we again obtain
\[
\mathrm{TAI}(w_2)<\mathrm{ASI}(w_2).
\]

\section{Relation to grammars and pattern-based formalisms}

The introduction of wildcard templates in $\mathcal{A}_{\ast}(w)$ connects templated Assembly Theory to established formalisms in theoretical computer science and formal language theory.

First, the canonical string-based assembly index is closely related to the smallest grammar problem: given a string $w$, find a smallest context-free grammar (or straight-line program) that generates exactly $w$.
In ~\cite{masierak_computational_2025} this correspondence is made precise, and NP-completeness of computing $\ASI(w)$ is established.

Templated assembly extends this picture in the direction of:
macro grammars and grammars with parameters, where non-terminals can take arguments and effectively encode templates with holes; pattern languages, where patterns are strings over an alphabet of terminals and variables, and variables can be instantiated by arbitrary strings; and insertion systems, where strings grow by insertion of substrings into designated positions.

Our block-compressed templates correspond to patterns whose literal segments are constrained to appear in the same order as in $w$, and whose variables (wildcards) represent compressible blocks between these anchors.

The templated assembly index $\TAI(w)$ can thus be viewed as an operational measure of how economically $w$ can be generated by a restricted class of pattern-based constructions that combine concatenation with a limited form of parameterised copying.
This perspective suggests both algorithmic strategies (e.g.\ approximate $\TAI(w)$ via grammar-based compression with macro rules) and complexity-theoretic results, to which we now turn.

\subsection{A greedy macro-grammar heuristic for approximating $\TAI(w)$}
The connection to macro grammars suggests practical heuristics for obtaining computable upper bounds on $\TAI(w)$.
Informally, one may treat a frequently instantiable pattern as a ``macro rule'' whose body is a template $T\in T(w)$ (the skeleton) and whose argument is the filler block $u$ used to instantiate wildcard sites via
$R_{\mathrm{temp}}$.

\paragraph{Candidate templates and filler families.}
Fix a template $T\in T(w)$.
For each $u\in O^{\ast}$ such that the fully instantiated string
\[
x_u := T[P(T)\mapsto u]
\]
is a substring of $w$ (hence $x_u\in \mathrm{Sub}(w)\subseteq O^{\ast}$), let $\mathrm{occ}(x_u)$ denote the number of occurrences
of $x_u$ in $w$ (or in the current working representation of $w$ maintained by the heuristic).
Instead of selecting a single pair $(T,u)$, it is natural to select one skeleton $T$ together with a family of fillers
$U\subseteq O^{\ast}$ for which $x_u\in \mathrm{Sub}(w)$.
This reflects the fact that the cost of constructing the skeleton $T$ (via $R_{\mathrm{concat}}$) can be amortised over multiple instantiations, possibly with different fillers.
In practice, it is often beneficial to restrict to templates with at least two wildcard sites, $|P(T)|\ge 2$, since a single instantiation step can then fill several sites at once (choose $S=P(T)$).

\paragraph{A simple gain score.}
To decide whether committing to a skeleton $T$ and a filler family $U$ is likely to reduce the number of steps, we compare:
(i) assembling all occurrences of the instantiated substrings $x_u$ ``directly'' in the canonical model, and
(ii) assembling the skeleton $T$ once (via $R_{\mathrm{concat}}$), assembling each filler $u\in U$ once, and then producing
each occurrence of $x_u$ by one instantiation in $R_{\mathrm{temp}}$.
Let $c(\cdot)$ be any efficiently computable cost proxy for canonical (concatenation-only) assembly complexity,
intended to approximate (or upper bound) $\ASI(\cdot)$ by taking classical reuse into account.
In the simplest baseline one may take $c(y):=|y|-1$.
Using such a proxy, define
\[
\mathrm{gain}(T,U):=\sum_{u\in U}\mathrm{occ}(x_u)\cdot c(x_u)
-\Bigl(c(T)+\sum_{u\in U}c(u)+\sum_{u\in U}\mathrm{occ}(x_u)\Bigr),
\qquad x_u:=T[P(T)\mapsto u].
\]
A positive score indicates that (under this proxy) it is worthwhile to pay for constructing $T$ (once) and the fillers in $U$
in order to amortise them over instantiations.
In a greedy implementation, one typically accepts only candidates with $\mathrm{gain}(T,U)>0$ and uses the magnitude of $\mathrm{gain}(T,U)$ as a priority score (larger gain indicates a more promising template under the chosen proxy).
Conversely, $\mathrm{gain}(T,U)<0$ suggests that the template is unlikely to be beneficial in this approximation, although one cannot exclude indirect benefits through later reuse effects in the full $\TAI$ model.

\paragraph{A sketch of a greedy heuristic.}
Starting from the explicit target word $w$, iterate:
\begin{enumerate}
  \item Mine candidate skeletons $T\in T(w)$ together with admissible fillers $u$ for which $x_u=T[P(T)\mapsto u]\in \mathrm{Sub}(w)$,
        and compute $\mathrm{occ}(x_u)$ in the current working representation.
  \item For each $T$, choose a filler family $U$ (e.g.\ all admissible fillers, or the best-scoring subset) and compute $\mathrm{gain}(T,U)$.
  \item Select a skeleton $T$ with a filler family $U$ of maximum positive gain.
  \item Commit to constructing $T$ once (using $R_{\mathrm{concat}}$), constructing each $u\in U$ once, and generating each chosen
        occurrence of $x_u$ via a single rule application in $R_{\mathrm{temp}}$ (fully parallel substitution, $S=P(T)$).
  \item ``Compress'' the working representation of $w$ by treating each generated occurrence of each $x_u$ as an available object,
        and continue searching for further profitable candidates on the remaining structure.
\end{enumerate}
The resulting set of committed constructions and instantiations can be translated directly into a templated assembly plan in $A^{\ast}(w)$, and therefore yields an explicit computable upper bound on $\TAI(w)$.
We leave a detailed experimental evaluation of such heuristics for future work.

\medskip
\noindent
This viewpoint is consistent with Example~3.2: the skeleton $T=\texttt{11*11*11}$ is instantiated in $w$ with two different fillers,
$u=\texttt{22}$ and $u=\texttt{00}$, producing $P=T[P(T)\mapsto \texttt{22}]$ and $S=T[P(T)\mapsto \texttt{00}]$.
For the baseline choice $c(y)=|y|-1$, we have $|T|=8$ and thus $c(T)=7$, while $|\texttt{22}|=|\texttt{00}|=2$ and hence
$c(\texttt{22})=c(\texttt{00})=1$.
Moreover, $|P|=|S|=10$, so $c(P)=c(S)=9$, and in this example $\mathrm{occ}(P)=\mathrm{occ}(S)=1$.
Taking $U=\{\texttt{22},\texttt{00}\}$, the aggregate score becomes
\[
\mathrm{gain}(T,U) = (1\cdot 9 + 1\cdot 9) - \bigl(7 + (1+1) + (1+1)\bigr) = 18 - 11 = 7,
\]
which is positive precisely because the one-time cost of constructing the skeleton $T$ is amortised across two instantiations with different fillers.
Accordingly, a gain-based greedy strategy would naturally commit to assembling the single skeleton $T$ once and then instantiating it for each filler, capturing the step savings exhibited by the templated assembly plan in that example.

paragraph{Gain calculations for the additional worked examples I and II.}
We reuse the baseline proxy $c(y):=|y|-1$.

\medskip\noindent {Example I ($w_1=10113121101011212111211$).}
Take the skeleton $T=1\ast 11$ and the filler family $U=\{0,2\}$, which produces
\[
x_0 = T[\ast\mapsto 0]=1011,\qquad x_2 = T[\ast\mapsto 2]=1211.
\]
Here $|T|=4$, so $c(T)=3$, and $|0|=|2|=1$, hence $c(0)=c(2)=0$.
Moreover $|x_0|=|x_2|=4$, so $c(x_0)=c(x_2)=3$.
In the block factorisation
\[
w_1 = 1011\;3\;1211\;0\;1011\;2\;1211\;1211
\]
we have $\mathrm{occ}(x_0)=2$ and $\mathrm{occ}(x_2)=3$. Therefore
\begin{align*}
\mathrm{gain}(T,U)
&= \bigl(2\cdot c(x_0)+3\cdot c(x_2)\bigr)
   - \Bigl(c(T)+\bigl(c(0)+c(2)\bigr)+\bigl(\mathrm{occ}(x_0)+\mathrm{occ}(x_2)\bigr)\Bigr)\\
&= (2\cdot 3+3\cdot 3)-\bigl(3+(0+0)+(2+3)\bigr)\\
&= 15-8 = 7.
\end{align*}

\medskip\noindent{Example II ($w_2=101131211010112121112111011$).}
With the same $T=1\ast 11$ and $U=\{0,2\}$, we still have
$c(T)=3$, $c(0)=c(2)=0$, and $c(1011)=c(1211)=3$.
Using the factorisation
\[
w_2 = 1011\;3\;1211\;0\;1011\;2\;1211\;1211\;1011,
\]
we obtain $\mathrm{occ}(1011)=3$ and $\mathrm{occ}(1211)=3$, hence
\begin{align*}
\mathrm{gain}(T,U)
&= (3\cdot 3+3\cdot 3)-\bigl(3+(0+0)+(3+3)\bigr)\\
&= 18-9 = 9.
\end{align*}

\section{Computational complexity}
\label{sec:complexity}

We consider the decision problem associated with the templated assembly index $\TAI(w)$.

\begin{definition}
The decision problem \textsc{TAI-DEC} is defined as follows:
\begin{itemize}
    \item Instance: a string $w\in\Sigma^{+}$ and a non-negative integer $k$ (given in unary or binary).
    \item Question: does there exist a templated assembly plan for $w$ of total cost at most $k$ in the templated assembly space
    $\mathcal{A}_{\ast}(w)$?
\end{itemize}
\end{definition}

\begin{lemma}
\label{lem:tai-in-np}
\textsc{TAI-DEC} belongs to $\mathrm{NP}$.
\end{lemma}

\begin{proof}[Sketch]
A certificate is a sequence of rule applications (concatenations and template instantiations), together with the intermediate objects they produce.
Since every step strictly increases the length of at least one constructed object, the length of a valid plan is polynomially bounded in $|w|+k$.
Verification proceeds by simulating each operation and checking that the produced object belongs to $O_{\ast}=\mathrm{Sub}(w)\cup \mathcal{T}(w)$.
All checks can be carried out in time polynomial in $|w|+k$, hence \textsc{TAI-DEC} is in $\mathrm{NP}$.
\end{proof}

\paragraph{On hardness.}
The canonical decision problem
\[
\textsc{ASI-DEC}=\{\, (w,k)\mid \ASI(w)\le k\,\}
\]
is NP-complete for fixed finite alphabets~\cite{masierak_computational_2025}.
Since $\TAI(w)\le \ASI(w)$ for all $w$, the templated model is a genuine extension of the canonical one.
However, this extension goes in the easier direction (it potentially shortens optimal plans), so NP-hardness of \textsc{TAI-DEC} does not follow immediately from NP-hardness of \textsc{ASI-DEC}.

One can nevertheless view \textsc{ASI-DEC} as the restriction of \textsc{TAI-DEC} to assembly plans that never apply a rule in $R_{\mathrm{temp}}$.
Equivalently, under the promise that $w$ admits no beneficial template instantiation (so that $\TAI(w)=\ASI(w)$), the problems \textsc{TAI-DEC} and \textsc{ASI-DEC} coincide.

This observation does not, however, settle the unconditional worst-case complexity of \textsc{TAI-DEC}.
Indeed, once beneficial template instantiations are allowed, deciding whether $\TAI(w)\le k$ requires
simultaneously reasoning about the classical reuse-vs.-concatenation trade-off and about whether
it is worthwhile to invest concatenation steps into constructing templates in order to amortize them
through one or more instantiation steps.
For this reason, we strongly conjecture that \textsc{TAI-DEC} remains NP-hard (already over fixed alphabets): the templated model preserves the combinatorial ``hard core'' of canonical assembly planning, while adding an additional layer of discrete structure selection (identifying and exploiting profitable templates).
While this is not a proof, it provides compelling evidence that intractability should persist in the worst case.
Establishing unconditional NP-hardness (or polynomial-time solvability) of \textsc{TAI-DEC} for fixed alphabets remains an interesting open question.\\

\section{Discussion}

We have proposed an extension of the canonical string-based formulation of Assembly Theory in which intermediate objects may contain a wildcard symbol $\ast$ representing compressed blocks, and templates are restricted to block-compressed substrings of the target string.
This leads to the templated assembly index $\TAI(w)$, a new complexity measure that generalises the canonical assembly index $\ASI(w)$ and is sensitive to templated modularity: the presence of modules that recur in multiple, possibly distant contexts linked by a common scaffold.\\

Formally, we defined templated assembly spaces, specified concatenation and template instantiation rules, and established basic relationships between $\TAI(w)$ and $\ASI(w)$.
We established membership in NP and clarified why NP-hardness does not follow from a trivial embedding argument.
We conjecture that \textsc{TAI-DEC} is NP-hard (already for fixed alphabets), while leaving unconditional NP-hardness open.
We also outlined connections to grammar-based compression and to formalisms such as pattern languages.\\

Some domains where $\TAI$ and its future variants may offer novel insights are:
\begin{itemize}
  \item Chemistry and molecular design: Many molecules can be described as relatively rigid scaffolds decorated by repeated functional groups: aromatic cores bearing several identical substituents, oligomers with the same side chain attached at multiple positions, or dendritic architectures built from a small set of branching units. In a string-based representation of molecules (for example a linearised encoding of a core with substituents), such patterns correspond to templates where wildcard positions mark sites of functionalisation. Canonical assembly indices quantify reuse of repeated fragments, but they do not distinguish between independent repetitions and repetitions coordinated by a shared backbone. In contrast, templated assembly makes this distinction explicit: a common scaffold with wildcard sites is built once, and then instantiated with chosen functional groups. Molecules or families of molecules with a large gap $\ASI(w) - \TAI(w)$ would then correspond to structures where much of the apparent combinatorial complexity comes from templated substitution patterns rather than from genuinely unrelated submotifs, suggesting applications in quantifying scaffold--decoration modularity in medicinal chemistry, combinatorial library design, and polymer or supramolecular engineering.
  \item Sequence analysis and genomics: Many genomic regions exhibit modular architectures, with recurring motifs (e.g.\ transcription factor binding sites, protein domains) embedded in diverse local contexts. Canonical assembly indices already provide a measure of hierarchical reuse, but templated indices could more directly reflect shared scaffolds with multiple occurrences of the same motif. Differential behaviour of $\ASI(w)$ and $\TAI(w)$ across genomic regions might highlight functionally relevant templated structures.
  \item Biosignatures and origins of life: Assembly Theory has been proposed as a framework for defining universal biosignatures based on complexity measures derived from putative assembly paces~\cite{cronin2017,cronin2022,cronin2023}. The templated assembly index offers a complementary axis: high $\TAI(w)$ together with a large gap $\ASI(w)-\TAI(w)$ might signal systems that exploit templated copying of modules in ways characteristic of evolved biological organisation, for example gene families under common regulatory architectures.\\
\end{itemize}

Future work may include a more detailed classification of the gap $\ASI(w)-\TAI(w)$ for various families of strings, refined complexity and approximation results, and empirical studies of $\TAI$ on biological and artificial datasets.
On the conceptual side, templated Assembly Theory suggests a broader programme: to systematically explore how different classes of templates and rewriting operations give rise to distinct families of assembly indices, each probing a different facet of structure and history in complex systems.

\bibliographystyle{vancouver}
\bibliography{main}

\end{document}